\documentclass[11pt]{article}
\usepackage{booktabs} 
\usepackage[ruled,vlined,linesnumbered,nofillcomment]{algorithm2e}

\SetAlFnt{\small}
\SetAlCapFnt{\small}
\SetAlCapNameFnt{\small}
\SetAlCapHSkip{0pt}
\IncMargin{-\parindent}

\usepackage{fullpage}
\usepackage{amsmath, amsthm, amssymb}
\usepackage{verbatim}
\usepackage[sort]{natbib}
\usepackage{xcolor}
\usepackage{graphicx}
\usepackage{hyperref}

\newtheorem{claim}{Claim}
\newtheorem{lemma}{Lemma}
\newtheorem{theorem}{Theorem}

\theoremstyle{definition}

\newtheorem{example}{Example}

\newcommand{\bE}{\mathbb{E}}
\newcommand{\bR}{\mathbb{R}}

\newcommand{\bN}{\mathbb{N}}
\newcommand{\eps}{\varepsilon}
\newcommand{\cA}{\mathcal{A}}
\newcommand{\cS}{\mathcal{S}}
\newcommand{\cH}{\mathcal{H}}
\newcommand{\trans}{P}
\newcommand{\si}{{s_\mathrm{init}}}
\newcommand{\st}{{s_\mathrm{term}}}
\newcommand{\ra}{r_A}
\newcommand{\rp}{r_P}
\newcommand{\ua}{u_A}
\newcommand{\up}{u_P}
\newcommand{\pf}{f}
\newcommand{\argmax}{\operatorname*{argmax}}

\newcommand{\dom}{D}
\newcommand{\dfp}{\delta_P}
\newcommand{\dfa}{\delta_A}

\renewcommand{\emptyset}{\varnothing}

\title{Efficient Algorithms for Planning with Participation Constraints}

\author{Hanrui Zhang\thanks{Carnegie Mellon University, \texttt{hanruiz1@cs.cmu.edu}.  Supported by NSF Award IIS-1814056, ARO Grant W911NF2110230, and funding from the Cooperative AI Foundation and the Center for Emerging Risk Research.}
\and
Yu Cheng\thanks{University of Illinois at Chicago, \texttt{yucheng2@uic.edu}.  Supported in part by NSF Award CCF-2122628.}
\and
Vincent Conitzer\thanks{Duke University, \texttt{conitzer@cs.duke.edu}.  Supported by NSF Award IIS-1814056, ARO Grant W911NF2110230, and funding from the Cooperative AI Foundation and the Center for Emerging Risk Research.}
}

\date{}

\begin{document}

\maketitle

\begin{abstract}
We consider the problem of {\em planning with participation constraints} introduced in~\citep{zhang2022planning}. In this problem, a principal chooses actions in a Markov decision process, resulting in separate utilities for the principal and the agent.  However, the agent can and will choose to end the process whenever his expected onward utility becomes negative.  The principal seeks to compute and commit to a policy that maximizes her expected utility, under the constraint that the agent should always want to continue participating. We provide the first polynomial-time exact algorithm for this problem for finite-horizon settings, where previously only an additive $\varepsilon$-approximation algorithm was known. Our approach can also be extended to the (discounted) infinite-horizon case, for which we give an algorithm that runs in time polynomial in the size of the input and $\log(1/\varepsilon)$, and returns a policy that is optimal up to an additive error of $\varepsilon$.
\end{abstract}

\section{Introduction}

{\em How do we keep users from leaving?}
That is the question asked daily by service providers such as banks, phone carriers, cable networks, and internet streaming companies.
Much of the depth of the question originates from its dynamic nature: services last over (normally an extensive period of) time, users can leave at almost any moment, and the cost and benefit of leaving vary depending on the situation.
Consider cable networks: when a new user signs their first contract, the network typically offers a discounted rate for 6 or 12 months, and if the user switches to another network during that time, there will be an early termination fee.
However, after the first several months, when the user has become attached to the network, the monthly rate increases to the normal amount.
Similar strategies (free trials, sign-up bonuses, etc.) are used by almost all service providers, especially those conducting business over the internet, where it is easier for users to leave a provider and switch to another.
While there are certainly many other considerations behind such strategies, arguably the main objective is to keep users around while generating as much revenue as possible.

Even in the simple example of cable networks, the dynamic nature of the problem already introduces some delicate tradeoffs: a low (or $0$) early termination fee would make it harder to keep the user in the first several months, but would also encourage the user to start the service in the first place (equivalently, prevent the user from leaving at the very beginning); similarly, a higher normal rate (which means higher revenue) may still be acceptable once the user becomes sufficiently attached, but anticipating this eventual higher normal rate at the outset, a new user may not sign the contract in the first place, or may leave before becoming attached to the network.
In other words, the network's policy in a ``state'' not only affects whether the user would leave in that state, but also affects the user's decision in all ``previous'' states.
Moreover, although the network and the user have misaligned interests, they are by no means in a zero-sum situation: the network may spend extra effort on improving the quality of the service, which would cost the network, but benefit the user even more, so they have less incentive to leave --- the question is, is that worthwhile?

The presence of these issues suggests that designing a business strategy should be viewed as a {\em planning} problem, where the network is the planner (or the {\em principal}), and the user is the {\em agent}.
The principal decides what {\em action} to take in each possible situation (i.e., each {\em state}).
The action gives the principal and the agent possibly different rewards, and brings the state to a possibly random new state, where another action will be taken.
The agent does not have a voice in which actions to take in which states, but always has the option to {\em leave}, which is the rational move to make when the agent's expected onward utility is below $0$ (where without loss of generality, $0$ is the utility induced by the best outside option, taking into consideration the cost of leaving).
The goal of the principal is to design a policy that maximizes the principal's utility {\em subject to participation constraints}, which require that the agent's expected onward utility in every possible state should be at least $0$.%
\footnote{Of course, sometimes it is more desirable for the principal to simply let the agent leave.  Still, technically, the assumption that the principal never allows the agent's expected onward utility to be negative is without loss of generality.
This is because we can extend the MDP with an ``end'' action from each state (where it is possible for the agent to have negative onward utility), which deterministically leads to 
an additional, absorbing state corresponding to the process having ended.  From this state, no additional rewards will be obtained by either party.  With these extensions, a policy that would result in the agent actually leaving the MDP at some point is equivalent to the policy that is the same except for, at that point, taking the ``end'' action within the MDP, ensuring zero onward utility.}
Such participation constraints introduce a mechanism design flavor to the problem, which distinguishes it from the classical problem of planning in Markov Decision Processes (MDPs).
In fact, the latter can be viewed as a special case of the former, where conceptually, the agent does not have the ability to leave. We can bring the classical case into the formalism here by making sure the agent's reward is always nonnegative, so the agent would never want to leave.

In this paper, we study the problem of planning with participation constraints from a computational point of view.
Our goal is to answer the following question:
\begin{quote}
    \em Given all parameters of a dynamic environment (i.e., reward functions and transition probabilities), can we efficiently compute a policy that maximizes the principal's utility subject to participation constraints?
\end{quote}

\subsection{Equivalent Variants}

For further motivation, we now present some variants that result in the same technical problem, so that our techniques apply to them as well.  The reader who is satisfied to keep the above motivation in mind can safely skip this subsection, as the remainder of the paper is written in line with the above motivation.

Equivalently, we can also consider problems where the goal is not to prevent the agent from {\em leaving}, but rather the goal is to prevent the agent from {\em entering}.  Such examples are reminiscent of problems in the security games literature \citep{kar2017trends,sinha2018stackelberg}.
For example, suppose we wish to discourage young people from joining a gang.  We consider a representative agent (young person) and assume that once he joins the gang, he can no longer leave it.  We can plan various, generally costly, enforcement measures that reduce the expected onward utility of being in the gang.  
(Note that we cannot {\em condition} these actions on whether the agent has joined the gang, as we are generally unable to observe gang membership; we have to commit to taking the actions even if we believe the agent was successfully deterred from joining the gang.)
In this case, the goal is to ensure that this expected onward utility always stays {\em nonpositive}, so the agent will not join the gang.  Simply negating the agent's rewards thus brings us back to the problem considered before.

In this example, our actions do not affect the agent's utility {\em before} the agent enters (joins the gang), whereas in the original problem we introduced, our actions do not affect the agent's utility {\em after} the agent leaves.
A natural generalization is that our actions may affect the agent's utility both before and after the agent has left or entered.  In the previous example, the most effective way to prevent the agent from joining the gang may not be to reduce the utility of being in the gang through enforcement, but rather to increase the utility of {\em not} being in the gang, for example by investing in after-school programs.  Or, perhaps a combination of both is optimal. In this case, there is no longer a sharp distinction between entering and leaving --- entering the gang is equivalent to leaving the alternative activities.  What matters is the {\em difference} in reward between having entered/left and not yet having done so.
If we normalize the rewards to the agent so that leaving results in onward rewards of $0$, we arrive back at the problem considered before.\footnote{Indeed, previously, when we assumed staying out of the gang gave rewards of $0$, and we then negated the rewards of being in the gang, this corresponded exactly to this normalization step: the negated rewards of being in the gang are the normalized rewards of staying out of the gang in that case.}
Thus, in the remainder of this paper, we focus on the original problem where leaving results in onward utility zero, in the understanding that this problem captures the full generality of such problems where rewards may be received both before and after leaving/entering.

\subsection{Our Results}

Our main result is an affirmative answer to our main question: {\em there is a polynomial-time algorithm that computes an optimal policy subject to participation constraints} (see Theorem~\ref{thm:main}).
Simple as it may appear, we find the existence of such an algorithm highly counterintuitive.
In classical MDPs, it is well known that optimal policies are without loss of generality deterministic and history-independent.
Given this, an optimal policy can be found by a simple backward induction procedure.
Unfortunately, this is no longer true in the presence of participation constraints.
In fact, as we show in Section~\ref{sec:difficulties}, restricting the policy to be either deterministic or history-independent may lead to an enormous loss in the principal's utility.
In other words, to solve our problem, we need to optimize over randomized and history-dependent policies.
Such optimization problems are often extremely hard (i.e., $\mathsf{APX}$-hard or $\mathsf{PSPACE}$-hard), which is the case for, e.g., partially observable MDPs and various special cases thereof \citep{papadimitriou1987complexity,mundhenk2000complexity}.
Another concrete example is that computing an optimal dynamic mechanism (a problem closely related to ours, which can be viewed as our setting with additional incentive-compatibility constraints) is $\mathsf{APX}$-hard~\citep{zhang2021automated}.
In fact, to the best of our knowledge, no other planning problem of a similar level of generality (i.e., generalizing planning in classical MDPs) where history-dependence is required admits efficient exact algorithms --- this phenomenon is famously known as the {\em curse of history} \citep{pineau2006anytime,silver2010monte,ye2017despot}.
Moreover, the fact that optimal policies may be history-dependent also rules out the possibility of computing the flat representation of an optimal policy efficiently, since the size of such a representation is already exponential in the number of states. (Our algorithm computes a succinct and implicit representation that encodes an optimal policy.)
Given all the above, at least we were surprised that an efficient algorithm exists for planning with participation constraints.

Technically, our algorithm operates over the concept of Pareto frontier curves (formally defined in Section~\ref{sec:curves}).
Roughly speaking, the Pareto frontier curve associated with a state specifies, for each given onward utility that we may wish to guarantee the agent, 
the maximum onward utility for the principal that is achievable by a policy that satisfies: (1) it gives the agent exactly the desired onward utility and (2) it satisfies all future participation constraints.
If we were able to somehow compute the Pareto frontier curves in all states, then it would be possible (although probably still nontrivial) to construct an optimal policy given these curves, or at least find the principal's optimal utility subject to participation constraints.
However, although these curves are piecewise linear, in general they have exponentially many pieces, which makes computing them explicitly impossible.
Our algorithm instead only tries to {\em evaluate} these curves in specific ways.
In particular, we make two types of evaluations: evaluations at specific points, and evaluations along specific directions.
While none of these evaluations can be done in a straightforward way (because we cannot compute the curves), we show that they can be recursively reduced to each other, through binary searching over the direction of an evaluation.
Then, by scheduling all recursive evaluations in the right order, the algorithm is able to perform all essential evaluations using only polynomial computation, {\em given that the binary searches only require polynomially many iterations}.
Bounding the number of iterations then requires a careful analysis of the numerical precision of the algorithm and the numerical ``resolution'' of the Pareto frontier curves, which turns out to work exactly in the way we want.
As a result, we obtain a weakly polynomial-time algorithm (similar to all currently known polynomial-time algorithms for linear programming), whose time complexity depends on the number of bits required to encode the input numbers.
A more detailed overview is given in Section~\ref{sec:overview}.

The algorithm discussed above is for finite-horizon (episodic) environments, but it is not too hard to adapt it into an algorithm for infinite-horizon discounted environments.
As a byproduct of our main result, we also give an algorithm that computes a policy that is additively suboptimal by at most $\eps$ for any $\eps > 0$ in infinite-horizon discounted environments, which runs in time polynomial in $\log(1 / \eps)$ and the size of the input.
This is discussed in Section~\ref{sec:extensions}, together with other remarks and extensions of the finite-horizon algorithm.

\subsection{Related Work}

Most closely related to our results is the recent work by \citet{zhang2022planning}.
They provide two algorithms for planning with participation constraints in finite-horizon environments: an approximation algorithm and an exact algorithm.
Their approximation algorithm computes a policy that can be additively suboptimal by at most $\eps$ for any $\eps > 0$, in time polynomial in $1 / \eps$, as well as the size of the problem.
Note that this guarantee is not only weaker compared to that provided by our exact algorithm, but also weaker than the one provided by our algorithm for infinite-horizon discounted environments, whose time complexity is polynomial in $\log(1 / \eps)$ rather than $1 / \eps$.
Their exact algorithm, which computes the Pareto frontier curves in all states, takes exponential time in the worst case.
Our polynomial-time exact algorithm closes the main question left open by \citet{zhang2022planning}, i.e., whether the problem of planning with participation constraints is in $\mathsf{P}$.

From an economic perspective, the problem of planning with participation constraints can be viewed as dynamic mechanism design (see, e.g., \citep{athey2013efficient,bergemann2010dynamic,bergemann2019dynamic,pavan2017dynamic,pavan2014dynamic}) under individual rationality constraints {\em only}.
The key difference is that in dynamic mechanism design, the agent has private information that may affect the reward of both the agent and the principal.
So, in addition to satisfying participation constraints (i.e., individual rationality constraints), normally the principal's policy also needs to be incentive compatible, so the agent is encouraged to report their private information truthfully.
From a computational point of view, the fact that the agent does not have private information enables polynomial-time algorithms for computing an optimal policy, which is known to be hard with incentive-compatibility constraints \citep{papadimitriou2016complexity,zhang2021automated}.
A conceptually related problem is that of ``moving the goalposts'' \citep{ely2020moving}, in which an agent works on a task of uncertain difficulty, modeled as the duration of required effort, and the principal knows the task difficulty and provides information over time, with the goal being to encourage the agent to finish the task.
This problem can be viewed as a structured special case of planning with participation constraints, where the state consists of the agent's belief of the difficulty of the task, as well as the fraction of the task that is already finished.
Another related problem is dynamic evaluation design \citep{smolin2021dynamic}, in which the principal evaluates an agent who is learning their own ability, with the goal being to persuade the agent that they are of high ability so that they will keep working.
Again, this problem can be viewed as a special case of ours, where the state is the agent's belief of their own ability.
These results are not comparable to ours, since they focus on {\em characterizing} optimal policies in {\em structured} environments, whereas our goal is to {\em compute} optimal policies in {\em general} environments.

The problem of planning with participation constraints is related to a number of planning problems in different variants and generalizations of MDPs.
In constrained MDPs (CMDPs) \citep{altman1999constrained}, the planner aims to find an optimal policy subject to an {\em overall} constraint, such as that the expected cumulative ``cost'' must be at most some certain amount.
It is known that in CMDPs, optimal policies are without loss of generality history-independent, and can be found by linear programming \citep{altman1996constrained,altman1998constrained,altman1995linear}.
Another related model is multi-objective MDPs (MOMDPs) \citep{roijers2013survey}.
Similar to CMDPs, MOMDPs focus on the overall cumulative reward vector, whereas in our problem, participation constraints have to be satisfied throughout the process.
In multi-agent (partially observable) MDPs \citep{gmytrasiewicz2005framework,hoang2013interactive,oliehoek2012decentralized}, multiple agents act individually in a common environment, based only on local information and beliefs about each other.
One key difference between our problem and multi-agent MDPs is that we consider an asymmetric environment where the principal has the exclusive power to choose a policy, and the agent can only choose to participate or not.

\section{Preliminaries}

We first formally introduce the problem setup, discuss why the problem is challenging, and introduce the notion of Pareto frontier curves which will be instrumental in the algorithm and the analysis thereof.

\subsection{Problem Setup}

\paragraph{The environment.}
We mostly focus on finite-horizon environments in this paper.
There are $n$ states $\cS = [n] = \{1, \ldots, n\}$ and $m$ actions $\cA$.\footnote{
    We assume all actions are available in every non-terminal state.
 This is without loss of generality because if an action $a$ is not available in a state $s$, we can set $a$'s rewards and transition probabilities to be the same as any available action in $s$.
}
For each state $s \in \cS$ and action $a \in \cA$, let $\rp(s, a)$ and $\ra(s, a)$ be the rewards of the principal and the agent respectively when action $a$ is played in state $s$.
Moreover, let $\trans(s, a) \in \bR^n$ be the transition probabilities when action $a$ is played in state $s$, where $\trans(s, a, s')$ is the probability that the next state is $s' \in \cS$.

Without loss of generality, we assume that the states are ordered by reachability. Formally, for any $s, s' \in \cS$ where $s \ge s'$, we have $\trans(s, a, s') = 0$ for all $a \in \cA$.\footnote{
    This is without loss of generality for finite-horizon (episodic) environments because one can make a copy of each state for each time step.
    Then, copies of states at earlier times can only transition into copies at later times.}
We assume $\si = 1$ is the initial state, and $\st = n$ is the terminal state where no action is available.

\paragraph{Histories and policies.}
A history of length $t \in \bN$ is a tuple $(s_1, a_1, \dots, s_t, a_t)$. 
Let $\cH_t$ be the set of all histories of lengths $t$ for each $t \in \bN$.
In particular, $\cH_0 = \{\emptyset\}$, where $\emptyset$ denotes the empty history.
Let $\cH = \bigcup_{t \in \bN} \cH_t$.
For history $h = (s_1, a_1, \dots, s_t, a_t) \in \cH$ and state-action pair $(s, a) \in \cS \times \cA$, we write $h + (s, a)$ for the history obtained by appending $(s, a)$ to the end of $h$:
\[
    h + (s, a) = (s_1, a_1, \dots, s_t, a_t, s, a).
\]
Define $(s, a) + h$ similarly.
For two history-state pairs $(h, s)$ and $(h', s')$ where $h = (s_1, a_1, \dots, s_t, a_t)$ and $h' = (s'_1, a'_1, \dots, s'_{t'}, a'_{t'})$, we say $(h', s')$ extends $(h, s)$, or $(h', s') \supseteq (h, s)$, if $t' > t$, and $(s_1, a_1, \dots, s_t, a_t, s)$ is a prefix of $(s'_1, a'_1, \dots, s'_{t'}, a'_{t'}, s')$.

Let $\Delta(\cA)$ denote the probability simplex over $\cA$.
A policy $\pi: \cH \times \cS \to \Delta(\cA)$ maps a history $h \in \cH$ and a state $s \in \cS$ to a random action $a \in \cA$, where $\pi(h, s, a)$ is the probability that $\pi$ plays action $a$ at history-state pair $(h, s)$.
Let $\Pi$ be the set of all (randomized, history-dependent) policies, which may or may not satisfy participation constraints (defined below).

\paragraph{Utility and participation constraints.}
Under a policy $\pi$, the expected onward utility $\up^\pi(h, s)$ of the principal at history-state pair $(h, s)$ can be defined in the following recursive way.
\begin{equation}
\label{eqn:u-p-pi}
    \up^\pi(h, s) = \begin{cases}
        0 & \text{if } s = \st, \\ 
        \bE_{a \sim \pi(h, s), s' \sim \trans(s, a)}[\rp(s, a) + \up^\pi(h + (s, a), s')] & \text{otherwise.}
    \end{cases}
\end{equation}
The onward utility of the agent $\ra^\pi(h, s)$ can be defined similarly, with $\up$ and $\rp$ replaced by $\ua$ and $\ra$ respectively.
We say a policy is feasible if it satisfies participation constraints in all states.
Throughout the paper, we assume that there exists a feasible policy (e.g., the policy that maximizes the agent's utility).
Our goal is to find a feasible policy that maximizes the principal's overall utility.
Formally, we want to compute a policy $\pi$ that maximizes $\up^\pi(\emptyset, \si)$, subject to the participation constraints that $\ua^\pi(h, s) \ge 0$ for all $(h, s) \in \cH \times \cS$.~\footnote{
    Note that some history-state pairs may not be reachable with positive probability under a policy.
    For consistency, we enforce participation constraints for such pairs as well.
    This is without loss of generality, since if $(h, s)$ is not reachable, then the policy from this point onward does not affect the principal's utility, so we can run the policy that maximizes the agent's utility to satisfy participation constraints.
}

\paragraph{Encoding the input.}
In order to properly formulate the computational problem, we assume that all parameters of the problem (including $n$, $m$, $\rp(s, a)$, $\ra(s, a)$, and $\trans(s, a, s')$) are given in binary representations.
Moreover, we assume that $-1 \le \rp(s, a) \le 1$ and $-1 \le \ra(s, a) \le 1$ for all $s$ and $a$, and each of the input numbers has at most $L$ bits.

\subsection{Some Natural Approaches and Why They Fail}
\label{sec:difficulties}

Before diving into our algorithm, we first discuss some natural approaches and why they do not work.
In classical MDPs, it is well known that optimal policies are without loss of generality deterministic and history-independent.
Given this, an optimal policy can be found by a simple backward induction procedure.
Unfortunately, this is no longer true in the presence of participation constraints.
As we illustrate in the following examples, restricting the policy to be either deterministic or history-independent may lead to a significant loss in the principal's utility.

\begin{figure}
\centering
\includegraphics[width=0.45\linewidth]{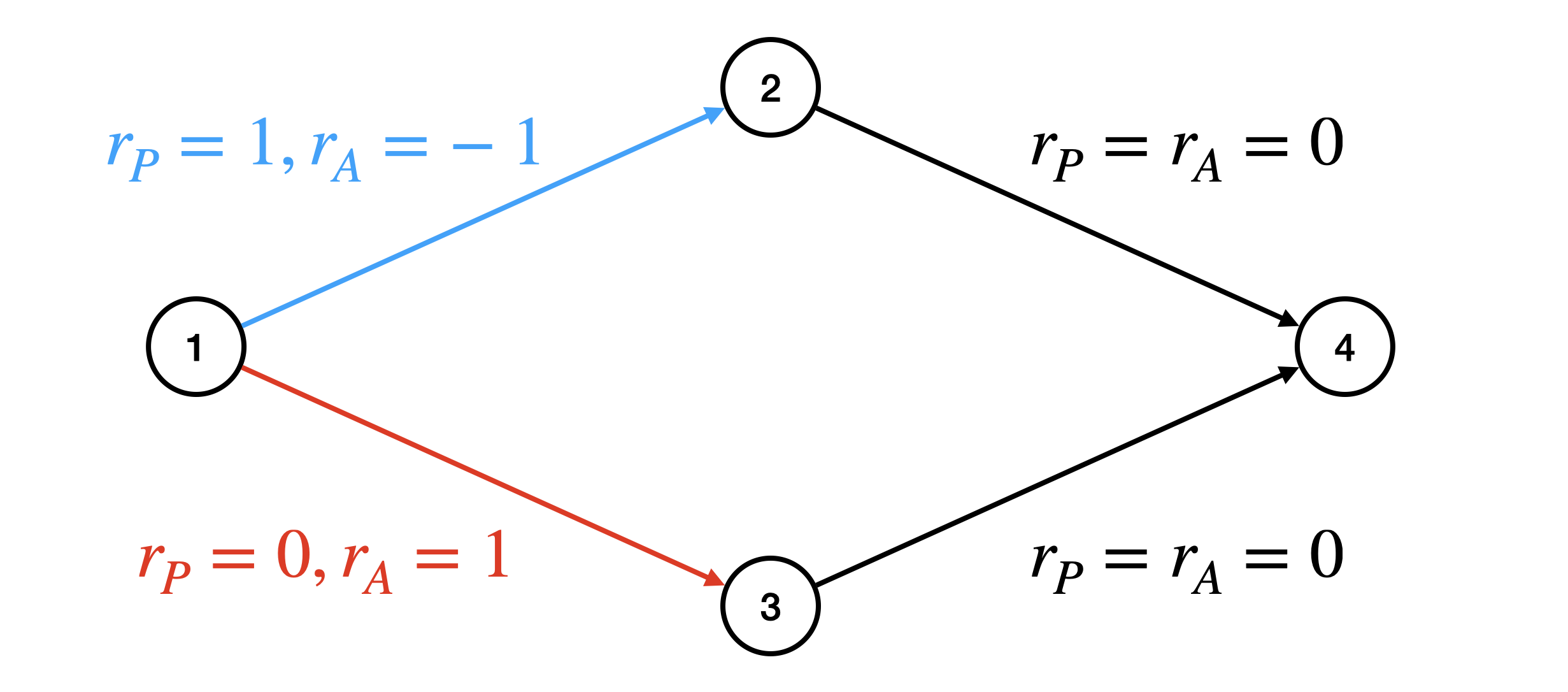}
\includegraphics[width=0.54\linewidth]{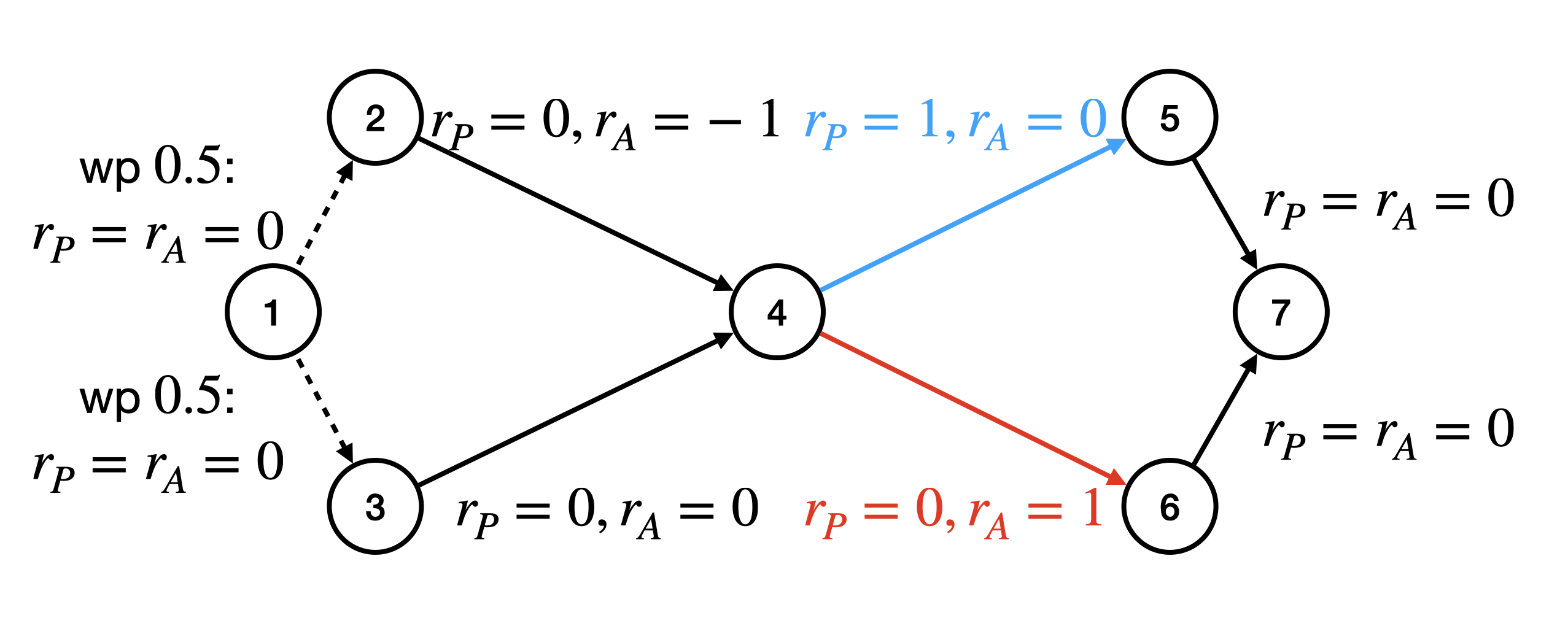}
\caption{Examples where deterministic/history-independent policies are far from optimal.}
\label{fig:examples}
\end{figure}

\begin{example}
    Consider the left environment in Figure~\ref{fig:examples}.
    This environment has $n = 4$ states, where $\si = 1$ and $\st = 4$.
    All states have at most $1$ available action except for state $1$.
    In state $1$ there are two actions available, the upper (blue) one and the lower (red) one, leading to state $2$ and state $3$ respectively.
    The optimal (randomized) policy is to play the upper action and the lower action each with probability $1/2$ in state $1$, which gives the principal overall utility $1/2$, and the agent onward utility $0$ in all states.
    However, restricted to deterministic policies, the only feasible policy is to play the lower action in state $1$, which gives the principal overall utility $0$.
\end{example}

\begin{example}
    Consider the right environment in Figure~\ref{fig:examples}.
    This environment has $n = 7$ states, where $\si = 1$ and $\st = 7$.
    All states have at most $1$ available action except for state $4$.
    In state $4$, there are two actions available, the upper (blue) one and the lower (red) one, leading to states $5$ and $6$ respectively.
    Moreover, in state $1$, the only available action randomly transits to state $2$ or $3$ with equal probability.
    The optimal (history-dependent) policy is to play the upper action in state $4$ if the previous state is state $3$, and play the lower action if the previous state is state $2$, which gives the principal overall utility $1/2$, and the agent nonnegative onward utility in all states.
    However, restricted to history-independent policies, the only feasible policy is to play the lower action in state $4$, which gives the principal overall utility $0$.
    In particular, note that in state $4$ we cannot play one of the two actions uniformly at random, because then the agent's onward utility in state $2$ would be $-1/2$.
\end{example}

Optimizing over history-dependent policies is often computationally intractable~\citep{pineau2006anytime,silver2010monte,ye2017despot}.
For instance, the problem of finding optimal policies for partially observable MDPs (as well as various special cases thereof~\citep{papadimitriou1987complexity,mundhenk2000complexity}) is PSPACE-hard.
Another concrete example is that computing an optimal dynamic mechanism (which can be viewed as our setting with additional incentive-compatibility constraints) is $\mathsf{APX}$-hard~\citep{zhang2021automated}.
Another difficulty that arises from history-dependence is that we cannot efficiently describe an optimal policy in the flat representation, since the optimal policy may need to specify which action to take in each of exponentially many histories.

We conclude this section by showing that it is computationally hard to find an optimal deterministic policy.
While the best deterministic policy could perform worse than the optimal randomized policy, there are situations where one may want to focus on deterministic policies.
More importantly, this further illustrates the complexity of our problem.
We reduce from the $0$-$1$ knapsack problem.

\begin{claim}
    It is $\mathsf{NP}$-hard to find an optimal deterministic policy that satisfies participation constraints.
\end{claim}
\begin{proof}
    Consider a knapsack instance with $k$ items and size limit $S$, where item $i$ has size $s_i$ and value $v_i$. The goal of the knapsack problem is to pick a subset of items with maximum total value, subject to the constraint that their total size does not exceed $S$.
    Without loss of generality, assume $s_i \le S$ for each $i \in [k]$.
    We construct an environment with $n = k+2$ states that encodes the knapsack instance, where $\si = 1$, $\st = n$, and state $i+1$ corresponds to item $i$.
    
    There is a single action $a_0$ available in state $\si = 1$ with $\rp(\si, a_0) = 0$, $\ra(\si, a_0) = \frac{1 - k}{k} \cdot S$, and $\trans(\si, a_0, i + 1) = \frac{1}{k}$ for each $i \in [k]$.
    For each item $i$, there are two actions $a_{i, 0}$, $a_{i, 1}$ available in the corresponding state $i + 1$.
    Intuitively, $a_{i, 0}$ corresponds to not taking item $i$, where $\rp(i + 1, a_{i, 0}) = 0$, $\ra(i + 1, a_{i, 0}) = S$; and $a_{i, 1}$ corresponds to taking item $i$, where $\rp(i + 1, a_{i, 1}) = v_i$, $\ra(i + 1, a_{i, 1}) = S - s_i$.
    Both actions lead to the terminal state deterministically, i.e., $\trans(i + 1, a_{i, 0}, \st) = \trans(i + 1, a_{i, 1}, \st) = 1$.
    
    We show that this encodes the knapsack instance.
    Due to the structure of the environment we construct, all policies are without loss of generality history-independent, so we omit the dependence on $h$.
    For a deterministic policy $\pi$, let $T^\pi \subseteq [k]$ be the set of items that $\pi$ decides to pick, that is, $i \in T^\pi$ iff $\pi(s, a_{i, 1}) = 1$.
    Then we have
    \begin{itemize}
        \item $\displaystyle
            \rp^\pi(\si) = \frac{1}{k} \sum_{i \in T^\pi} v_i$. \\
        \item For each $i \in [k]$, we always have $\ra^\pi(i + 1) \ge 0$.
        \item $\displaystyle
            \ra^\pi(\si) \ge 0 \iff \sum_{i \in T^\pi} s_i \le S$.
    \end{itemize}
    It follows immediately that an optimal deterministic policy subject to participation constraints corresponds to an optimal solution to the knapsack instance.
\end{proof}

\subsection{Pareto Frontier Curves}
\label{sec:curves}

Now we define the notion of Pareto frontier curves, which is instrumental in designing and analyzing our algorithm.
Intuitively, these curves capture the Pareto optimal tradeoffs between the principal's and the agent's (onward) utilities at different states.

We associate a Pareto frontier curve with each state $s \in \cS$.
For state $s$, we consider all policies starting at $s$ (as if $s$ is the initial state) and the onward utilities of the principal and the agent $\ua^\pi$ and $\up^\pi$ as defined in Equation~\eqref{eqn:u-p-pi}.
We say a policy $\pi$ is {\em feasible in the future} iff $\pi$ satisfies the participation constraints at all later history-state pairs.

Let $\dom_s = [\ua^-(s), \ua^+(s)]$ be the range of onward utility of the agent that is achievable by policies that are feasible in the future.
Formally,
\begin{align*}
    \ua^-(s) & = \min\{\ua^\pi(\emptyset, s) \mid \pi \in \Pi: \ua^\pi(h', s') \ge 0,\, \forall (h', s') \supseteq (\emptyset, s)\}, \\
    \ua^+(s) & = \max\{\ua^\pi(\emptyset, s) \mid \pi \in \Pi: \ua^\pi(h', s') \ge 0,\, \forall (h', s') \supseteq (\emptyset, s)\}.
\end{align*}
Note that we consider policies that satisfy participation constraints {\em after leaving} state $s$, and put no restrictions on the agent's onward utility {\em in} state $s$.

The Pareto frontier curve $\pf_s: \dom_s \to \bR$ in state $s \in \cS$ maps the agent's onward utility $x \in \dom_s$ to the maximum principal's onward utility $y$ that is achievable by some feasible-in-the-future policy $\pi$, such that the agent's onward utility is exactly $x$ under $\pi$.
Formally, for each $s \in \cS$ and $x \in \dom_s$,
\begin{equation}
\label{eqn:pareto-curve}
    \pf_s(x) = \max\{\up^\pi(\emptyset, s) \mid \pi \in \Pi: \ua^\pi(\emptyset, s) = x \text{ and } \ua^\pi(h', s') \ge 0,\, \forall (h', s') \supseteq (\emptyset, s)\}.
\end{equation}

The following property of Pareto frontier curves, which was observed in \citep{zhang2022planning}, plays an important role in our algorithm and analysis.

\begin{lemma}
\label{lem:concave}
    For each $s \in \cS$, the Pareto frontier curve $\pf_s$ defined in Equation~\eqref{eqn:pareto-curve} is concave on $\dom_s$.
\end{lemma}

The concavity of these curves is a direct consequence of the fact that randomizing between feasible-in-the-future policies always results in a feasible-in-the-future policy.

\section{Our Algorithm and Analysis}

Our main result is a polynomial-time exact algorithm for the problem of planning with participation constraints.

\begin{theorem}
\label{thm:main}
    There is an algorithm that runs in time $\mathrm{poly}(n, m, L)$ and computes an optimal policy satisfying participation constraints, where $n$, $m$, and $L$ are the number of states, number of actions, and number bits required to encode each input number (rewards and transition probabilities) respectively.
\end{theorem}

The proof of the theorem is deferred to Section~\ref{sec:proof}, and the next subsection is dedicated to a more friendly presentation of the algorithm and the analysis.

\subsection{Overview of the Algorithm}
\label{sec:overview}

Given the definition of Pareto frontier curves, the maximum overall utility of the principal that can be achieved by a feasible policy is equal to $\max_{x \in \dom_\si \cap \bR_+} \pf_\si(x)$.~\footnote{We use $\bR^+$ to denote the set of nonnegative real numbers.}
So, the problem of planning with participation constraints immediately reduces to computing the Pareto frontier curve at the initial state $\si$ --- which, unfortuantely, turns out to be a highly challenging (if not impossible) task.
In particular, although each $\pf_s$ is piecewise linear, there may be exponentially many pieces in each curve, which makes explicitly computing the curves infeasible.
In \cite{zhang2022planning}, the authors circumvent this issue by allowing approximation --- they give an approximation algorithm (which achieves an additive $\eps$-approximation in $\mathrm{poly}(1 / \eps)$ time) for planning with participation constraints by recursively computing approximations of the Pareto frontier curves, from later states to earlier ones.
Their main technical contribution is identifying a computationally feasible recursive relation between the curves, and coming up with a way to approximate the curves using only a small number of pieces.
However, it seems unlikely that similar approaches could lead to an efficient {\em exact} algorithm.

\begin{figure}
\centering
\includegraphics[width=0.48\linewidth]{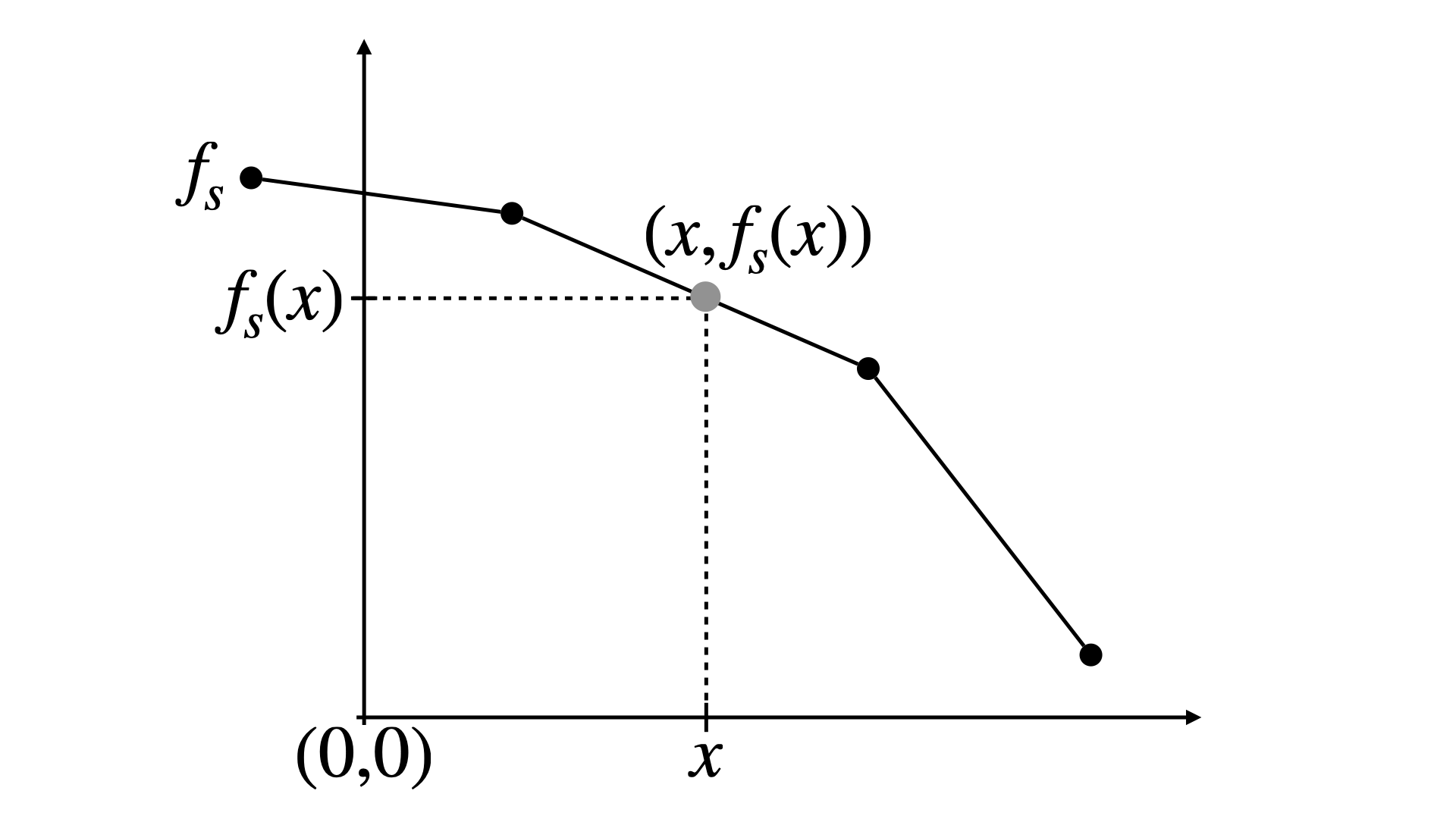}
\includegraphics[width=0.48\linewidth]{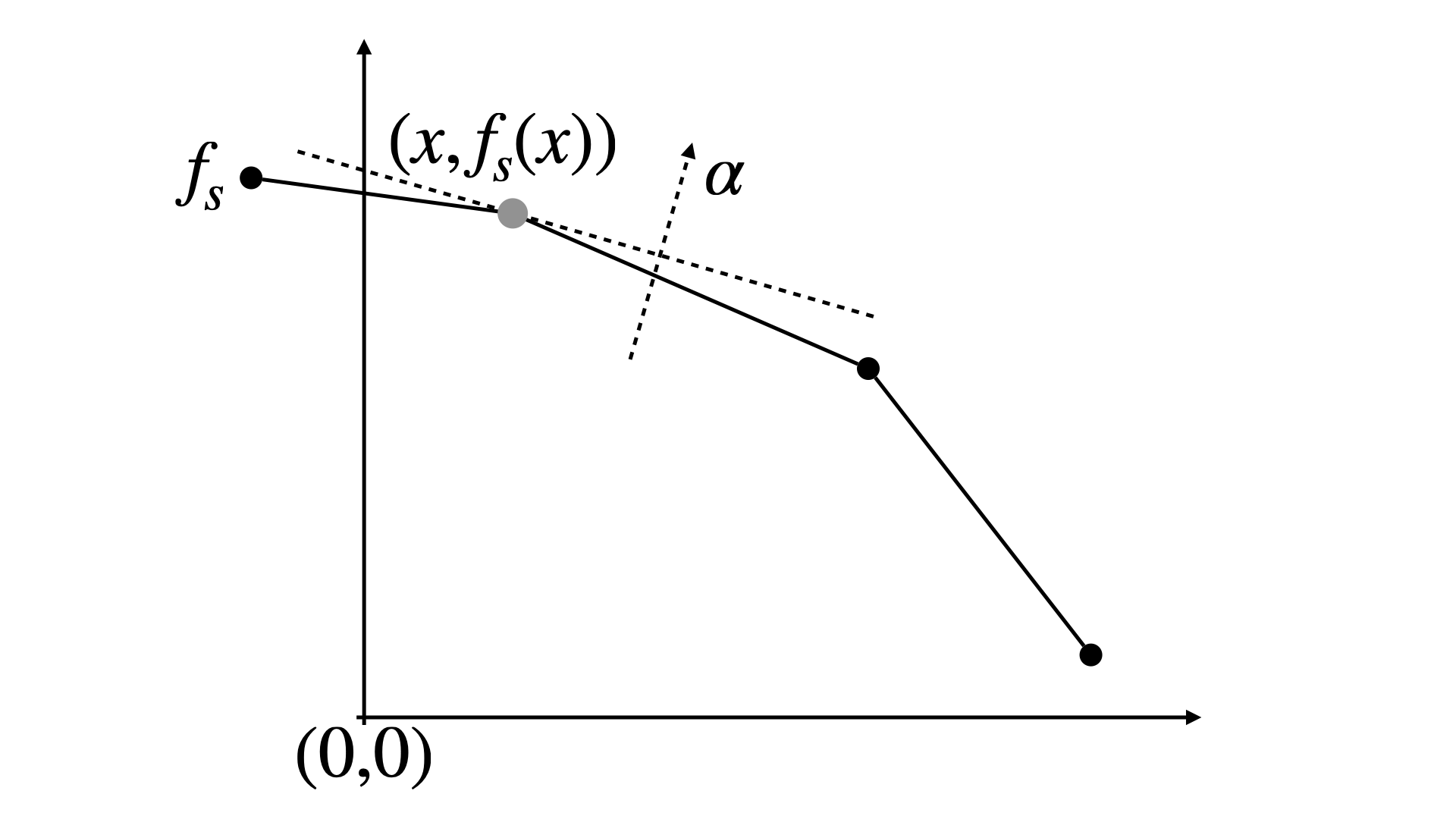}
\caption{The two types of evaluation subroutines of our algorithm.}
\label{fig:evaluation}
\end{figure}

In contrast to their approach, our algorithm does not try to compute (or approximate) the entire Pareto frontier curves.
Instead, we only evaluate the curves ``at specific points'' and ``along specific directions'' (see Figure~\ref{fig:evaluation}).
The left side of Figure~\ref{fig:evaluation} illustrates evaluating $\pf_s$ at a given point, where we want to compute $\pf_s(x)$ for a given $x$.
The right side of Figure~\ref{fig:evaluation} shows an evaluation along a specific direction $\alpha \in \bR^2$, which returns a point $(x, \pf_s(x))$ that maximizes the inner product $\alpha \cdot (x, \pf_s(x))$.
These two types of evaluations correspond to the two major {\em conceptual} subroutines of our algorithm.

If these subroutines can be implemented efficiently, then we can immediately compute the maximum overall utility of the principal: it is equal to the $y$-coordinate of the point found by evaluating $\pf_\si$ along the direction $(0, 1)$ if the $x$-coordinate of the returned point is nonnegative; otherwise it is equal to $\pf_\si(0)$. 
This is true because $\pf_\si$ is concave: if a point with the largest $y$-coordinate on $\pf_\si$ is to the left of $x = 0$, then the optimal feasible point must have $x$-coordinate $0$.

Based on these observations, we only need to efficiently implement these two subroutines.
Below we discuss how this is possible.
We will refrain from being fully formal, and focus on the intuition instead.
For the full description of the algorithm, see Algorithm~\ref{alg:main}.

\paragraph{Evaluations at specific points.}
Suppose we want to evaluate $\pf_s$ at $x$.
We show that this can be reduced to multiple evaluations along specific directions of $\pf_s$.
Consider the left side of Figure~\ref{fig:subroutines}.
To find the (gray) point $(x, \pf_s(x))$, we only need to find the two endpoints of the piece containing it, namely the (blue) point $(x_1, \pf_s(x_1))$ and the (green) point $(x_2, \pf_s(x_2))$.
Then, taking the convex combination of these two endpoints with the right coefficients gives us $(x, \pf_s(x))$.
These coefficients can be computed using $x$ (given as input), $x_1$ and $x_2$, as illustrated in Figure~\ref{fig:subroutines}.
So it suffices to find the two endpoints.

Consider, for example, the left endpoint $(x_1, \pf_s(x_1))$.
There exists some direction $\alpha$ (e.g., $\alpha_1$ in the figure) such that
\[
    \alpha \cdot (x_1, \pf_s(x_1)) = \max_{x \in \dom_s} \alpha \cdot (x, \pf_s(x)).
\]
We only need to find such an $\alpha$ and evaluate $\pf_s$ along that direction.
To this end, observe that the maximizer found by evaluating along $\alpha$ moves on the curve monotonically as we rotate $\alpha$ (consider, from the left to the right, $\alpha_3, \alpha_1, \alpha_2$ and the corresponding maximizers, which are the red, blue, and green points respectively).
Again this is because the curve is concave.
So, we need to find the ``rightmost'' $\alpha$ such that the maximizer found by evaluating along $\alpha$ is to the left of $x$, i.e., the $x$-coordinate of that maximizer is no larger than $x$.

To achieve this, we perform a binary search over $\alpha$.
We defer the discussion on the numerical issues of this binary search to Section~\ref{sec:numerical}.
For now, we assume the number of iterations this binary search requires is $\mathrm{poly}(n, m, L)$, which is in fact the case, as we will show later.

\paragraph{Evaluations along specific directions.}
Now consider the other subroutine where we want to evaluate $\pf_s$ along a given direction $\alpha$.
We show that this can be reduced to multiple evaluations of both types, of the Pareto frontier curves in later states.
At a high level, evaluating $\pf_s$ along $\alpha$ can be viewed as a planning problem, where the goal is to find a policy $\pi$ that maximizes $\alpha \cdot (\up^\pi(\emptyset, s), \ua^\pi(\emptyset, s))$, subject to participation constraints {\em in the future} (and not in state $s$).

Since the policy is unconstrained in state $s$, without loss of generality, an optimal policy $\pi$ has the Markovian property {\em in state $s$ only}: consider the behavior of the policy right after taking action $a$ in $s$, leaving $s$, and entering a later state $s' > s$.
The subpolicy from this point on must maximize $\alpha \cdot (\up^\pi((s, a), s'), \ua^\pi((s, a), s'))$ subject to participation constraints (including in state $s'$).
In particular, the subpolicy at $s'$ does not depend on the action $a$ taken in state $s$ or the subpolicy in other later states.\footnote{
    Note that the subpolicy in state $s'$ is in effect only if $s'$ is the first state reached after leaving $s$.
    In the case where we reach some $s''$ immediately after leaving $s$ and then later reach $s'$, it is the subpolicy in $s''$ that should apply.
}
This subpolicy corresponds to a point on $\pf_{s'}$, which can be found by evaluating $\pf_{s'}$ twice: along direction $\alpha$ and at $x = 0$ respectively, and then picking the point with the larger $x$-coordinate (again because $\pf_{s'}$ is concave).
In other words, the planning subproblem in each state $s' > s$ can be reduced to two evaluations of $\pf_{s'}$.
After solving these subproblems for each $s' > s$, the policy in state $s$ should choose an action $a$ which maximizes
\[
     \alpha \cdot \left(\rp(s, a), \ra(s, a)\right) + \bE_{s' \sim \trans(s, a)}\left[\alpha \cdot \left(\up^\pi((s, a), s'), \ua^\pi((s, a), s')\right)\right].
\]
The above procedure is illustrated in the right side of Figure~\ref{fig:subroutines}, where the action $a$ is a maximizer of the above expectation.
The two cases inside the expectation in the figure correspond to the two cases of the subproblem in each later state $s'$.
The upper case is when evaluating $\pf_{s'}$ along $\alpha$ returns the (gray) point $(x', \pf_{s'}(x'))$ with $x' \ge 0$, so it corresponds to the subpolicy at $s'$.
The lower case is when evaluating $\pf_{s'}$ along $\alpha$ gives a point with a negative $x$-coordinate, in which case the (gray) point $(0, \pf_{s'}(0))$ corresponds to the subpolicy in state $s'$.

\begin{figure}
\centering
\includegraphics[width=0.49\linewidth]{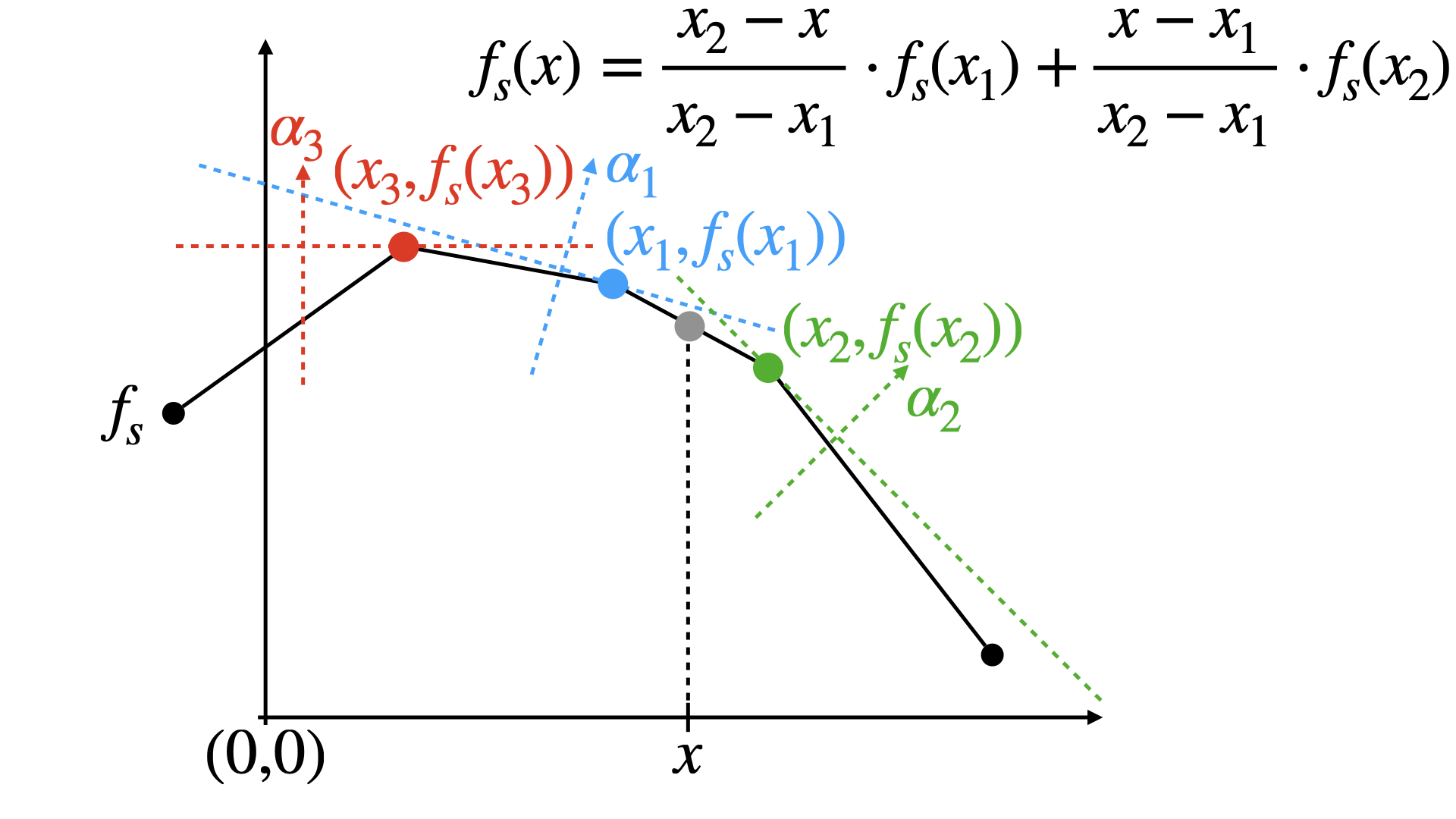}
\includegraphics[width=0.49\linewidth]{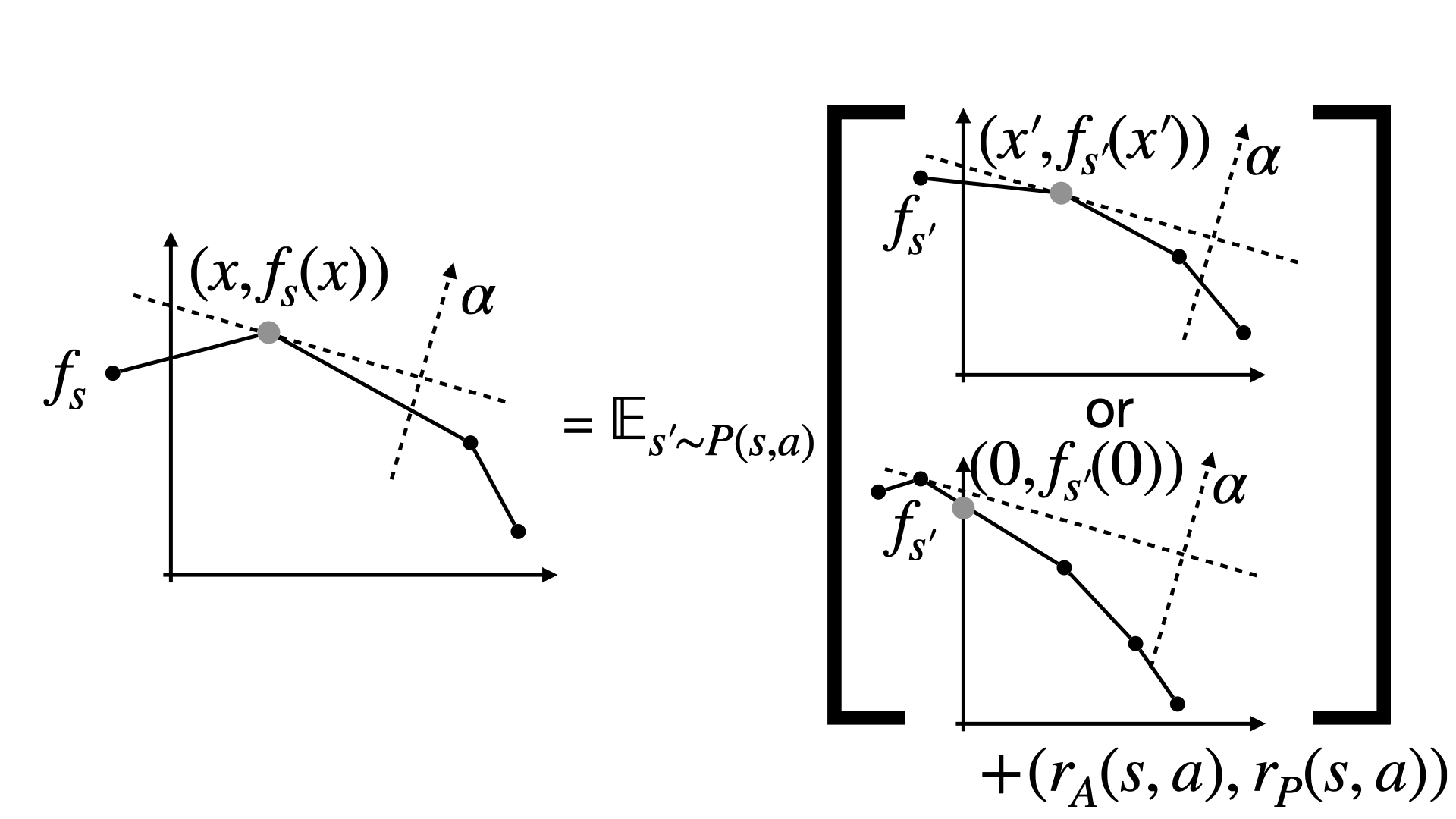}
\caption{Recursive (na\"ive) implementations of the two subroutines.}
\label{fig:subroutines}
\end{figure}

\paragraph{Putting everything together.}
The above discussion already describes a way to perform both types of evaluations in {\em finite} time.
This is because evaluations at specific points reduce to only evaluations along specific directions in the same state; and evaluations along specific directions reduce to only evaluations in later states (one evaluation of each type for each later state).
However, a problem is that it generally takes exponential time if we recursively perform the evaluations for subproblems in the na\"ive way, because
for both types of evaluations there can be polynomially many subproblems (i.e., the number of iterations in the binary search for evaluations at specific points, and the number of later states for evaluations along specific directions).

So, to evaluate $\pf_\si$ at $x = 0$ and along $(0, 1)$ efficiently, we need to schedule and handle all the evaluations involved (most of which originate from recursive calls) in a more global manner.
Intuitively, for any state $s \in \cS$, we only ever need to calculate $\pf_s(0)$ and evaluate $\pf_s$ along a polynomial number of directions $\alpha$.
This is because each direction $\alpha$ can be traced back to one of $n$ states that first queries for this $\alpha$; at the same time, each state only queries a polynomial number of different $\alpha$'s.
Below we give an informal hierarchical description of the schedule, together with inline annotations.
\begin{itemize}
    \item First observe that evaluations at $x = 0$ appear repeatedly in the na\"ive implementation.
    We therefore center our schedule around these evaluations.
    \item We will compute $\pf_s(0)$ for all states $s \in \cS$ one by one from later states to earlier ones (i.e., from $\st = n$ to $\si = 1$), since the recursive dependence (as discussed above) never goes backwards.
    We call this the {\bf outer loop}.
    \begin{itemize}
        \item Consider some state $s$ in the outer loop, and suppose we have already computed $\pf_{s'}(0)$ for all $s' > s$.
        As discussed above, to compute $\pf_s(0)$, it suffices to perform a {\bf binary search} on $\alpha$ in state $s$.
        \begin{itemize}
            \item In each iteration of the binary search, we need to perform an evaluation of $\pf_s$ along $\alpha$.
            As discussed above, we only need to evaluate $\pf_{s'}$ along $\alpha$ for each $s' > s$, since we already know $\pf_{s'}(0)$.
            This can be done in a single backward pass (from $\st = n$ to $s + 1$) without nested recursive calls, which we call the {\bf inner loop}.
            \begin{itemize}
                \item For each $s' > s$ in the inner loop, the evaluation of $\pf_{s'}$ along $\alpha$ reduces to the evaluation of $\pf_{s''}$ along $\alpha$ and $\pf_{s''}(0)$ for all $s'' > s'$.
                \item The former has already been computed in previous iterations of the inner loop, and the latter has already been computed in previous iterations of the outer loop.
                We only need to retrieve the two points for $s'' = s'+1, \ldots, n$, which means every iteration of the inner loop takes $O(n)$ time.
            \end{itemize}
            \item The inner loop has $O(n)$ iterations, so the total time is $O(n^2)$, which is also the runtime of one iteration of binary search.
        \end{itemize}
        \item Now as discussed before, the binary search has $\mathrm{poly}(n, m, L)$ iterations (we will come back to this not-yet-substantiated claim momentarily), so the total time is $\mathrm{poly}(n, m, L)$.
    \end{itemize}
    \item The outer loop has $O(n)$ iterations, so the total time of computing $\pf_s(0)$ for all $s \in \cS$ is $\mathrm{poly}(n, m, L)$.
    \item Finally, we need one last evaluation of $\pf_\si$ along the direction $(0, 1)$.
    This can be done by a single call to the inner loop above, which takes time $O(n^2)$.
\end{itemize}
A formal description of our algorithm is given in Algorithm~\ref{alg:main}.

\begin{algorithm}[!ht]
\KwIn{state space $\cS = [n]$, action space $\cA$, reward functions $\rp$ and $\ra$, and transition probabilities $\trans$.}
\KwOut{an implicit representation of a principal-optimal policy subject to participation constraints.}
    \tcc{the outer loop}
    \For{$s = n, n - 1, \dots, 1$}{
        \tcc{binary search for the endpoints of the piece containing $(0, \pf_s(0))$}
        let $\ell \gets 0$, $r \gets 2^{3nL}$\;
        \While{$r - \ell \ge 2^{-5n^2L}$}{
            let $\alpha \gets ((r + \ell) / 2, 1) \in \bR^2$ \;
            \tcc{the inner loop}
            \For{$s' = n, n - 1, \dots, s$}{
                let
                    $a_{s', \alpha} \gets \argmax_{a \in \cA} \alpha \cdot \left((\ra(s', a), \rp(s', a)) + \bE_{s'' \sim \trans(s', a)}[(x_{s'', \alpha}, y_{s'', \alpha})]\right)$
                \;
                let
                    $(x_{s', \alpha}, y_{s', \alpha}) \gets (\ra(s', a_{s', \alpha}), \rp(s', a_{s', \alpha})) + \bE_{s'' \sim \trans(s', a_{s', \alpha})}[(x_{s'', \alpha}, y_{s'', \alpha})]$
                \;
                \tcc{replace $(x_{s', \alpha}, y_{s', \alpha})$ with $(0, \pf_{s'}(0))$ if $x_{s', \alpha} < 0$; this is possible only for $s' > s$, where $\pf_{s'}(0) = y_{s'}$ has already been computed}
                \If{$x_{s', \alpha} < 0$ and $s' > s$}{
                    let $(x_{s', \alpha}, y_{s', \alpha}) \gets (0, y_{s'})$\;
                }
            }
            let $\ell \gets \alpha$ if $x_{s, \alpha} \le 0$, and $r \gets \alpha$ otherwise\;
        }
        let $\alpha_{s, -} \gets (\ell, 1)$, $\alpha_{s, +} \gets (r, 1)$\;
        let $(x_{s, -}, y_{s, -}) \gets (x_{s, \alpha_{s, -}}, y_{s, \alpha_{s, -}})$, $(x_{s, +}, y_{s, +}) \gets (x_{s, \alpha_{s, +}}, y_{s, \alpha_{s, +}})$\;
        \tcc{compute $y_s = \pf_s(0)$ as a linear combination of $y_{s, -}$ and $y_{s, +}$}
        let $y_s \gets (x_{s, +} \cdot y_{s, -} - x_{s, -} \cdot y_{s, +}) / (x_{s, +} - x_{s, -})$\;
        \tcc{fix infeasible points reached during the binary search}
        \For{each $\alpha$ tried in the above binary search where $x_{s, \alpha} < 0$}{
            let $(x_{s, \alpha}, y_{s, \alpha}) \gets (0, y_s)$\;
        }
    }
    let $e_y = (0, 1)$\;
    \For{$s = n, n - 1, \dots, 1$}{
        let $a_{s, e_y} \gets \argmax_{a \in \cA} e_y \cdot \left((\ra(s, a), \rp(s, a)) + \bE_{s' \sim \trans(s, a)}[(x_{s', e_y}, y_{s', e_y})]\right)$\;
        let $(x_{s, e_y}, y_{s, e_y}) \gets (\ra(s, a_{s, e_y}), \rp(s, a_{s, e_y})) + \bE_{s' \sim \trans(s, a)}[(x_{s', e_y}, y_{s', e_y})]$\;
        \tcc{this time we do not need to handle $\si = 1$ separately}
        \If{$x_{s, e_y} < 0$}{
            let $(x_{s, e_y}, y_{s, e_y}) \gets (0, y_s)$\;
        }
    }
    \tcc{the principal's optimal reward is $y_{\si, e_y} = y_{1, e_y}$}
    \Return all $\{(x_{s, -}, x_{s, +})\}$, $\{x_{s, \alpha}\}$, $\{(\alpha_{s, -}, \alpha_{s, +})\}$, and $\{a_{s, \alpha}\}$ computed above\;
\caption{A polynomial-time algorithm for computing a principal-optimal policy subject to participation constraints.}
\label{alg:main}
\end{algorithm}

\subsection{Handling Numerical Issues}
\label{sec:numerical}
Finally, we come back to the number of iterations required in the binary search in an evaluation at a specific point (used in the algorithm to compute $\pf_s(0)$ for each state $s$).
We show that under appropriate parametrization of the direction $\alpha$, it suffices to perform the binary search up to some singly exponential precision, which implies the number of iterations is polynomial.
In particular, we search over the slope of the perpendicular direction to $\alpha$.
The intuition is that we only need to distinguish between the slopes of two consecutive pieces on $\pf_s$, which cannot be too close to each other.
In fact, we will establish a stronger claim: the coordinates of all turning points on $\pf_s$ must be integral multiples of some singly-exponentially small quantity.
Since these coordinates are bounded between $-n$ and $n$, the slope of the line between any two turning points (which do not even need to be adjacent) cannot take too many values.
Moreover, the magnitude of the slope is upper bounded by some not too large quantity.
This allows the binary search to terminate in not too many steps.
We elaborate in the following paragraph.

To see why the above is true, we consider a specific procedure of recursively constructing the entire $\pf_s$ for all $s$ from later states to earlier ones, and treat all quantities involved in the construction as fractions.
Fix some $s$, and assume $\pf_{s'}$ for any $s' > s$ has the desired property, i.e., the denominator of any quantity used to represent $\pf_{s'}$ is not too large.
We argue that $\pf_s$ also has this property (where the denominator may be {\em moderately} larger than the denominators in later states --- in fact, it is the blowup that we try to bound).

Recall that any turning point on $\pf_s$ can be found by evaluating along some direction.
So fix a turning point $(x, y)$ on $\pf_s$, and consider any direction $\alpha$ which gives this point.
As discussed earlier (see the right side of Figure~\ref{fig:subroutines}), there exists some action $a \in \cA$ such that
\begin{align*}
    (x, y) & = \bE_{s' \sim \trans(s, a)}[(x_{s'}, y_{s'})] + (\ra(s, a), \rp(s, a)) \\
    & = \sum_{s' > s} \trans(s, a, s') \cdot (x_{s'}, y_{s'}) + (\ra(s, a), \rp(s, a)),
\end{align*}
where $(x_{s'}, y_{s'})$ is either a turning point on $\pf_{s'}$ or $(0, \pf_{s'}(0))$.
Among all quantities on the right hand side of the above equation, $\trans(s, a, s')$, $\ra(s, a)$, and $\rp(s, a)$ have at most $L$ bits in the binary representation, so the denominators of these quantities are at most $2^L$.
Moreover, when $(x_{s'}, y_{s'})$ is a turning point on $\pf_{s'}$, by the induction hypothesis, the denominators of both coordinates are not too large.
So if all $(x_{s'}, y_{s'})$ are turning points, then we immediately know that $(x, y)$ has the desired property: the denominator of both coordinates can only blow up by a factor of $2^L$.
And importantly, in any case, the denominator of $x_s$, or $x_{s'}$ for any $s' > s$, can be at most $2^{nL}$, because the $x$-coordinates of turning points in any state only depend on the $x$-coordinates of turning points in later states.

The problematic case is when $(x_{s'}, y_{s'}) = (0, \pf_{s'}(0))$.
To handle this case, we need to argue that the denominator of $y_{s'}$ is not too large either, compared to those of the turning points on $\pf_{s'}$.
Recall that there exist turning points $(x_\ell, y_\ell)$ and $(x_r, y_r)$ on $\pf_{s'}$, where $x_\ell < 0$ and $x_r > 0$, such that
\[
    y_{s'} = \frac{x_r \cdot y_\ell - x_\ell \cdot y_r}{x_r - x_\ell}.
\]
So the denominator of $y_{s'}$ (compared to that of $y_\ell$ or $y_r$) can blow up by at most a factor of $2^{nL}$ (which is the maximum denominator of $x_\ell$ and $x_r$), times the numerator of $x_r - x_\ell$ (which is upper bounded by $2n \cdot 2^{nL}$ because $-n \le x_\ell < x_r \le n$).
So in any case, assuming $n \ge 2$, the maximum blowup incurred in the construction of $\pf_s$ is $2^{3nL}$, and consequently, the denominator of the $y$-coordinate of any turning point on the curve of any state is $2^{3n^2L}$.
The above discussion on numerical issues is formalized as Lemma~\ref{lem:precision}, which is stated and proved in Section~\ref{sec:proof}.

\subsection{Decoding the Policy}

Algorithm~\ref{alg:main} outputs only an implicit representation of an optimal policy.
In this subsection, we describe how to efficiently decode the output of Algorithm~\ref{alg:main}, so that we can execute the corresponding policy.
The idea is to keep track of the current ``objective direction'', which is initially $(0, 1)$ and may randomly change (and in particular, rotate to the right) as the state evolves.
This objective direction essentially corresponds to how much we need to compensate the agent from this point on in order to satisfy participation constraints.
In other words, the objective direction succinctly encodes the relevant part of the history.

According to Algorithm~\ref{alg:main}, at any time, the onward policy in the current state (given the history) corresponds to either the maximizer along the objective direction, or the intersection of the Pareto frontier curve with $x = 0$. 
In the former case, the agent is satisfied with the current level of compensation, so we do not need to compensate more.
In this case, we can take an action deterministically, and the objective direction does not change.
In the latter case, we need to compensate the agent more to satisfy participation constraints, so we rotate the objective direction to the right.
In this case, we need to randomize between the two actions corresponding to the two endpoints of the piece containing the intersection point.
Depending on which action is actually chosen, the new objective direction is the one for which the corresponding endpoint is the maximizer.
Since all these points and directions (along with many other auxiliary points and directions) have been computed in Algorithm~\ref{alg:main}, we only need to read them from the output.
The full algorithm is given as Algorithm~\ref{alg:decode}, which maps each history-state pair to a (random) action.
Algorithm~\ref{alg:decode} can be easily adapted into a dynamic procedure that plays the optimal policy on the fly while interacting with the environment (by updating $\alpha$), rather than re-analyzing the history from scratch at each point in time.

\begin{algorithm}[ht]
\KwIn{the output of Algorithm~\ref{alg:main} and a history-state pair $(h, s)$ where $h = (s_1, a_1, \dots, s_t, a_t)$.}
\KwOut{a possibly random action corresponding to the optimal policy found by Algorithm~\ref{alg:main}.}

    let $\alpha \gets e_y = (0, 1)$\;
    \tcc{trace the history and compute the current internal state of the policy}
    \For{$i = 1, 2, \dots, t$}{
        \If{$x_{s_i, \alpha} = 0$}{
            if $a_i = a_{s_i, \alpha_{s_i, -}}$ then let $\alpha \gets \alpha_{s_i, -}$; otherwise let $\alpha \gets \alpha_{s_i, +}$\;
        }
    }

    \If{$x_{s, \alpha} = 0$}{
        \tcc{optimal onward policy corresponds to point $(0, \pf_s(0))$, which requires randomization in state $s$}
        \Return $a_{s, \alpha_{s, -}}$ with probability $x_{s, +} / (x_{s, +} - x_{s, -})$, and $a_{s, \alpha_{s, +}}$ with probability $-x_{s, -} / (x_{s, +} - x_{s, -})$\;
    } \Else {
        \tcc{optimal onward policy corresponds to the maximizer along direction $\alpha$}
        \Return $a_{s, \alpha}$\;
    }

\caption{An polynomial-time algorithm for decoding and executing (one step of) the optimal policy found by Algorithm~\ref{alg:main}.}
\label{alg:decode}
\end{algorithm}

Note that the behavior of the policy output by Algorithm~\ref{alg:main} is unspecified for some history-state pairs.
However, if one strictly follows the specified part of the policy, then the unspecified part can never be reached (i.e., the probability that we arrive at such a history-state pair is $0$).
For such unreachable pairs, the behavior of the decoding algorithm can be arbitrary.

\subsection{Remarks and Extensions}
\label{sec:extensions}

\paragraph{Structure of optimal policies.}
Our algorithm also directly implies some structural properties of optimal policies with participation constraints.
In particular:
\begin{itemize}
    \item Although there might be exponentially many turning points on the Pareto frontier curves, for the optimal policy we compute, there are only polynomially many of them for which it is possible that the policy visits them.
    These points are maximizers for the polynomially many directions we queried during the computation of an optimal policy.
    \item Optimal policies are almost deterministic.
    In fact, an optimal policy randomizes between precisely $2$ actions when the participation constraint in the current state (given the history) is binding.
    This is also where the policy branches and history-dependence is introduced.
    In all other situations, the policy deterministically chooses an action.
    This aligns well with intuition: when no participation constraints are binding, it suffices to simply maximize the principal's utility, which naturally leads to a completely deterministic and history-independent policy.
\end{itemize}

\paragraph{Extensions to richer constraints.}
Our algorithm can be generalized to the case where the agent's onward utility in each state must be in one of several disjoint intervals (instead of a single interval $[0, \infty)$).
Moreover, these feasible intervals can be different for each state.
In order to handle multiple feasible intervals in a state $s$, we evaluate $\pf_s$ at the endpoints of all these intervals, which can be done by binary search.
Once we have computed these points, to evaluate a curve along a direction $\alpha$, we only need to handle subproblems where we evaluate later curves {\em restricted to feasible intervals}.
This can be done since if the unconstrained maximizer is infeasible, then the optimal feasible point must be one of the two endpoints that are closest to the unconstrained maximizer.
Since we have already computed all endpoints, we can simply try the two points and choose the better one.

\paragraph{Infinite-horizon environments with discounted reward.}
Now we discuss how to extend our algorithm to the infinite-horizon case with discounted reward.
For such environments, we describe an algorithm that computes a policy subject to participation constraints that is optimal up to an additive error of $\eps > 0$, in time $\mathrm{poly}(n, m, L, \log(1 / \eps))$ for any $\eps > 0$.
This is done by reducing to the finite-horizon case and running Algorithm~\ref{alg:main}.

First we briefly define infinite-horizon environments.
As in the finite-state case, there are $n$ states $\cS = [n]$ and $m$ actions $\cA$, and $\rp$, $\ra$ and $\trans$ denote the principal's reward, the agent's reward, and the transition probabilities respectively.
There is an initial state $\si = 1$, but no terminal state.
We also do not require transitions to be from earlier states to later ones.
In addition, there are discount factors $\dfp \in (0, 1)$ and $\dfa \in (0, 1)$ (which we treat as constants) for the principal and the agent respectively.
Define histories similarly as in finite-horizon environments.
The onward utility $\up^\pi(h, s)$ of the principal under policy $\pi$ in state $s$ given history $h$ is defined recursively such that
\[
    \up^\pi(h, s) = \bE_{a \sim \pi(h, s), s' \sim \trans(s, a)}[\rp(s, a) + \dfp \cdot \up^\pi(h + (s, a), s')].
\]
The onward utility $\ua^\pi$ of the agent is defined similarly, with $\up^\pi$ and $\rp$ replaced with $\ua^\pi$ and $\ra$.
Participation constraints require that for all $(h, s) \in \cH \times \cS$, $\ua^\pi(h, s) \ge 0$.
We say a policy $\pi$ is feasible if it satisfies participation constraints.
The goal is to find a feasible policy that maximizes the principal's overall utility $\up^\pi(\emptyset, \si)$.

Note that the finite-horizon case can be viewed as a special case of the infinite-horizon case with discounted reward, by scaling the rewards appropriately and replacing the terminal state with an absorbing state from which there are no more rewards.
As a result, optimal policies in the infinite-horizon case in general also need to be randomized and history-dependent, so traditional methods are unlikely to work for the problem.
This is true even for approximately optimal policies, as illustrated in the examples in Section~\ref{sec:difficulties}.
Therefore, to handle the infinite-horizon case, it is necessary to incorporate the ideas developed in our algorithm for the finite-horizon case.

Our algorithm consists of two parts.
Based on the principal's discount factor $\dfp$ and the desired accuracy $\eps$, we first compute a cutoff time
\[
    T = O(\log(1 / (\eps \cdot (1 - \dfp))) / \log(1 / \dfp)).
\]
The idea is that the contribution to the overall utility after the first $T$ stages is at most $\frac{1}{1 - \dfp} \cdot \dfp^T \le \eps$.
After the $T$-th stage, we run a stationary policy that is optimal for the agent, which can be computed in polynomial time (through linear programming, or any other algorithm for computing optimal policies for standard infinite-horizon MDPs with discounted rewards).
Then we treat the first $T$ stages as a finite-horizon environment, and run Algorithm~\ref{alg:main} for this environment.

Since $T = O(\log(1 / \eps))$, this blows up the size of the problem at most by a $O(\log(1 / \eps))$ factor (as we make a copy of every state for every period).
Two aspects of how this finite-horizon version is set up deserve mention.  First, to match the infinite-horizon version, we have to discount the rewards in this finite-horizon version.  This is a straightforward modification: as every state in the finite-horizon version is already indexed by time, we can simply adjust the rewards for those time-indexed states by the appropriate discount factors.  Second, we still have to account for the discounted utility that the agent receives after $T$, as this may make it easier to satisfy the participation constraints before $T$.  To do so, we can simply add the total expected discounted utility after $T$ (from the agent-optimal stationary policy) as a single lump-sum reward to the final non-terminal state in the finite version.
The overall policy is then to run the output of Algorithm~\ref{alg:main} in the first $T$ stages, and to run the agent-optimal stationary policy after the $T$-th stage.

To see why this policy is only suboptimal by at most $\eps$, observe that the expected discounted principal utility that it obtains {\em from the first $T$ stages} is at least the expected discounted principal utility that the overall-optimal policy obtains from those stages.  This is because Algorithm~\ref{alg:main} explicitly optimizes for the first $T$ stages only, and the participation constraints it faces in these first $T$ stages cannot be tighter than those faced by the optimal policy, as the participation constraints for the finite-horizon version correspond to being as generous as possible to the agent after $T$.
Furthermore, the expected discounted principal utility that our algorithm obtains from the stages after $T$ can be at most $\eps$ lower than that for the optimal policy, by our choice of $T$.

\subsection{Proof of Theorem~\ref{thm:main}}
\label{sec:proof}

In this section, we present the proof of our main result (Theorem~\ref{thm:main}).
We start by proving several key technical lemmas.
We first prove the following lemma, which provides a tractable interpretation of evaluations along specific directions.

\begin{lemma}
\label{lem:evaluation}
    For any state $s \in \cS$ and direction $\alpha \in (\bR \times \bR_+)$,
    \[
        \max_{x \in \dom_s} \alpha \cdot (x, \pf_s(x)) = \max_{a \in \cA} \left(\alpha \cdot (\ra(s, a), \rp(s, a)) + \bE_{s' \sim \trans(s, a)}\left[\max_{x' \in \dom_{s'} \cap \bR_+} \alpha \cdot (x', \pf_{s'}(x'))\right]\right).
    \]
\end{lemma}
\begin{proof}
    We first show the left hand side is greater than or equal to the right hand side.
    Let $a^*$ and $x_{s'} \ge 0$ for each $s' > s$ be the maximizers on the right hand side.
    By the definition of $\pf_{s'}$, each $(x_{s'}, \pf_{s'}(x_{s'}))$ corresponds to a subpolicy $\pi^*_{s'}$ starting from state $s'$.
    We have
    \[
        (x_{s'}, \pf_{s'}(x_{s'})) = (\ua^{\pi^*_{s'}}(\emptyset, s'), \up^{\pi^*_{s'}}(\emptyset, s')).
    \]
    Moreover, for any $(h, s'') \supseteq (\emptyset, s')$,
        $\ua^{\pi^*_{s'}}(h, s'') \ge 0$.
    Now consider the policy $\pi$ defined such that $\pi(\emptyset, s) = a^*$, and for each $h = (s, a^*, s', a_2, s_3, \dots, s_t, a_t)$ and $s'' \in \cS$,
    \[
        \pi(h, s'') = \pi^*_{s'}((s', a_2, s_3, \dots, s_t, a_t), s'').
    \]
    That is, $\pi$ follows the recommendations of $\pi^*_{s'}$ whenever the first state reached after leaving $s$ is $s'$.
    For any unspecified history-state pair, $\pi$ always maximizes the agent's utility.
    It is easy to show that
    \[
        (\ua^\pi(\emptyset, s), \up^\pi(\emptyset, s)) = (\ra(s, a^*), \rp(s, a^*)) + \bE_{s' \sim \trans(s, a)}[(x_{s'}, \pf_{s'}(x_{s'}))].
    \]
    And moreover, because each $\pi^*_{s'}$ is feasible in the future and $x_{s'} \ge 0$, $\ua^\pi(h, s'') \ge 0$ for any $(h, s'') \supseteq (\emptyset, s)$.
    This means
    \begin{align*}
        \max_{x \in \dom_s} \alpha \cdot (x, \pf_s(x)) & \ge \alpha \cdot (\ua^\pi(\emptyset, s), \up^\pi(\emptyset, s)) \\
        & = \alpha \cdot (\ra(s, a^*), \rp(s, a^*)) + \alpha \cdot \bE_{s' \sim \trans(s, a)}[(x_{s'}, \pf_{s'}(x_{s'}))].
    \end{align*}

    Now consider the other direction.
    Let $x^*$ be the maximizer on the left hand side, and $\pi^*$ be the corresponding policy.
    Without loss of generality, $\pi^*(\emptyset, s) = a^*$ is deterministic (because otherwise we can simply choose the best action in the support).
    For each $s'$, let $\pi_{s'}$ be such that
    \[
        \pi_{s'}(h, s'') = \pi((s, a^*) + h, s'').
    \]
    That is, $\pi_{s'}$ is the subpolicy starting from $s'$ induced by $\pi^*$.
    Then because $\pi^*$ is feasible in the future, each $\pi_{s'}$ is also feasible in the future, and moreover, $\ra^{\pi_{s'}}(\emptyset, s') \ge 0$.
    So we have:
    \begin{align*}
        \alpha \cdot (x^*, \pf_s(x^*)) & = \alpha \cdot (\ra(s, a^*), \rp(s, a^*)) + \bE_{s' \sim \trans(s, a^*)}[\alpha \cdot (\ra^{\pi_{s'}}(\emptyset, s'), \rp^{\pi_{s'}}(\emptyset, s'))] \\
        & \le \max_{a \in \cA} \left(\alpha \cdot (\ra(s, a), \rp(s, a)) + \bE_{s' \sim \trans(s, a)}\left[\max_{x' \in \dom_{s'} \cap \bR_+} \alpha \cdot (x', \pf_{s'}(x'))\right]\right).
    \end{align*}
    This concludes the proof.
\end{proof}

The next lemma states that the denominators of the $x$ and $y$ coordinates returned by any evaluation of any Pareto curve throughout the algorithm are never too large, which is useful for upper bounding the number of iterations of binary search.

\begin{lemma}
\label{lem:precision}
    Consider all coordinates as fractions.
    Then we have: 
        (1) the least common denominator of the $x$-coordinates of the turning points on $\{\pf_s\}_{s \in \cS}$ is at most $2^{nL}$, and
        (2) the least common denominator of both the $x$-coordinates and the $y$-coordinates of the turning points on $\{\pf_s\}_{s \in \cS}$ is at most $2^{3n^2L}$.
\end{lemma}
\begin{proof}
    We start by proving the first statement by mathematical induction.
    For $\st = n$, each turning point on $\pf_\st$ is $(\ra(\st, a), \rp(\st, a))$ for some $a \in \cA$, and since each $\ra(\st, a)$ has at most $L$ bits, $2^L$ is a denominator of the $x$-coordinate of each turning point.
    
    Now fix some $s < \st = n$ and suppose for any $s' > s$ and any turning point on $\pf_{s'}$, $2^{(n - s)L}$ is a denominator of the $x$-coordinate of that point.
    We argue that for any turning point on $\pf_s$, $2^{(n - s + 1)L}$ is a denominator of the $x$-coordinate of the point.
    Consider any turning point $(x, \pf_s(x))$.
    Observe that there is a direction $\alpha \in \bR \times \bR_+$ such that
    \[
        \alpha \cdot (x, \pf_s(x)) = \max_{x' \in \dom_s} \alpha \cdot (x', \pf_s(x')).
    \]
    So by Lemma~\ref{lem:evaluation}, there exists an action $a \in \cA$ and some $x_{s'} \in \dom_s \cap \bR_+$, such that
    \[
        x = \ra(s, a) + \sum_{s' > s} \trans(s, a, s') \cdot x_{s'}.
    \]
    Moreover, since $x_{s'}$ is a maximizer, without loss of generality, either $x_{s'} = 0$ or $x_{s'}$ is a turning point on $\pf_{s'}$.
    In both cases, by the induction hypothesis, $2^{(n - s)L}$ is a denominator of $x_{s'}$.
    Since $2^L$ is a denominator of both $\ra(s, a)$ and $\trans(s, a, s')$, $2^{(n - s + 1)L}$ must be a denominator of $x$.
    This establishes the first half of the lemma.
    
    Now consider the second statement.
    We inductively show that for any state $s \in \cS$, we can use $2^{3(n - s + 1)nL}$ to upper bound some common denominator of both coordinates of all points on $\pf_{s'}$, as well as $\pf_{s'}(0)$, for all $s' \ge s$.
    Given the first half of the lemma, we only need to argue about the $y$-coordinates.
    
    First consider $\st = n$.
    For the turning points, each $\rp(\st, a)$ has at most $L$ bits, and $2^L$ is a denominator.
    As for $\pf_\st(0)$, let $(x_-, y_-)$ and $(x_+, y_+)$ be the endpoints of the piece containing $(0, \pf_\st(0))$ on $\pf_\st$.
    Observe that
    \[
        \pf_\st(0) = \frac{y_- \cdot x_+ - y_+ \cdot x_-}{x_+ - x_-}.
    \]
    So the product of the denominator of $y_- \cdot x_+ - y_+ \cdot x_-$ and the numerator of $x_+ - x_-$ is a denominator of $\pf_\st(0)$.
    The former is at most $2^{2L}$, and the latter is at most $2 \times 2^L$, so something no larger than $2^{3L + 1} \le 2^{3nL}$ is a common denominator of all the $y$-coordinates.
    
    Now suppose for all $s' > s$, some $D \le 2^{3(n - s)nL}$ is a denominator of all the $y$-coordinates used to represent all $\pf_{s'}$ (including all turning points and $\pf_{s'}(0)$).
    We first argue that some $2^L \cdot D$ is a common denominator of all the $y$-coordinates used to represent all $\pf_{s'}$ for all $s' \ge s$, excluding $\pf_s(0)$ (we will handle $\pf_s(0)$ separately).
    Fix a turning point $(x, \pf_s(x))$, and again consider a direction $\alpha \in \bR \times \bR_+$ such that $(x, \pf_s(x))$ is a maximizer.
    By Lemma~\ref{lem:evaluation}, there exists $a \in \cA$ and $x_{s'} \in \dom_s \cap \bR_+$ such that
    \[
        \pf_s(x) = \rp(s, a) + \sum_{s' > s} \trans(s, a, s') \cdot \pf_{s'}(x_{s'}).
    \]
    And each $x_{s'}$ is either a turning point or $0$.
    By the induction hypothesis, $2^L \cdot D$ is a denominator of all $\pf_s(x)$ where $(x, \pf_s(x))$ is a turning point.
    Finally consider $\pf_s(0)$.
    Again, let $(x_-, y_-)$ and $(x_+, y_+)$ be the endpoints of the piece containing $(0, \pf_s(0))$ on $\pf_s$.
    Observe that
    \[
        \pf_s(0) = \frac{y_- \cdot x_+ - y_+ \cdot x_-}{x_+ - x_-}.
    \]
    So the product of the denominator of $y_- \cdot x_+ - y_+ \cdot x_-$ and the numerator of $x_+ - x_-$ is a denominator of $\pf_s(0)$.
    The former, as discussed above, is at most $2^{nL} \cdot 2^L \cdot D \le 2^{L + nL + 3(n - s)nL}$, and the latter is at most $2n \times 2^{nL} \le 2^{nL + n}$ (because the denominator of $x_+ - x_-$ is at most $2^{nL}$, and $x_+ - x_- \le 2n$), so there exists a number that is at most $2^{3(n - s)nL + 2nL + n + L} \le 2^{3(n - s + 1)nL}$ as a common denominator of all the $y$-coordinates that we care about.
    This finishes the proof of the lemma.
\end{proof}

Now we are ready to prove the correctness of Algorithm~\ref{alg:main}.

\begin{proof}[Proof of Theorem~\ref{thm:main}]
    As discussed in the overview in Section~\ref{sec:overview}, Algorithm~\ref{alg:main} runs in time $\mathrm{poly}(n, m, L)$.
    We focus on proving the correctness of  Algorithm~\ref{alg:main}.
    
    In particular, for each pair $(s, \alpha)$ reached in the execution of the algorithm, $(x_{s, \alpha}, y_{s, \alpha})$ satisfies
    \[
        \alpha \cdot (x_{s, \alpha}, y_{s, \alpha}) = \max_{x \in \dom_s \cap \bR_+} \alpha \cdot (x, \pf_s(x)).
    \]
    Moreover, for each $s \in \cS$, $y_s = \pf_s(0)$.
    The claim regarding $(x_{s, \alpha}, y_{s, \alpha})$ can be proved inductively.
    In particular, for those points computed in the inner loop (lines~7~and~9) where $s' > s$, the property of $(x_{s', \alpha}, y_{s', \alpha})$ follows from the same property of each $(x_{s'', \alpha}, y_{s'', \alpha})$ and Lemma~\ref{lem:evaluation}.
    As for $(x_{s, \alpha}, y_{s, \alpha})$, the only difference is that when it is first computed in line~7, it is possible that $x_{s, \alpha} < 0$.
    However, this is fixed in line~15 given that $y_s = \pf_s(0)$.

    To show $y_s = \pf_s(0)$, we only need to show that the binary search is accurate enough.
    In particular, $(x_{s, -}, y_{s, -})$ and $(x_{s, +}, y_{s, +})$ are in fact the two endpoints of the piece containing $(0, \pf_s(0))$.
    Suppose that this is not the case.
    That is, without loss of generality, suppose there exists a turning point $(x, y)$ to the right of $(x_{s, -}, y_{s, -})$ where $x \le 0$.
    Let $\alpha = (t, 1)$ be a direction for which $(x, y)$ is the maximizer.
    It must be the case that $\alpha_{s, +}$ is to the right of $\alpha$, which is to the right of $\alpha_{s, -}$.
    In other words, at line~11, it must be the case that $\ell < t < r$.
    Consider the slopes of the piece containing $(0, \pf_s(0))$, and the piece immediately to the left of that piece, and let $k_1$ and $k_2$ be the two slopes respectively where $k_1 > k_2$.
    We must have $-r \le k_2 \le -t \le k_1 \le -\ell$, which in particular implies that $r - \ell \ge k_1 - k_2$.
    Now by Lemma~\ref{lem:precision}, the least common denominators of the two coordinates of all turning points are at most $2^{nL}$ and $2^{3n^2L}$ respectively.
    Moreover, all $x$-coordinates are between $-n$ and $n$.
    So, the minimum possible difference between the slopes of two consecutive pieces is at least $1 / (2n \cdot 2^{nL} \cdot 2^{3n^2L}) \ge 2^{-5n^2L}$.
    This means $r - \ell \ge k_1 - k_2 \ge 2^{-5n^2L}$, which contradicts the stopping criterion of the binary search (line~3).

    One final concern is that the initial $r$ (line~2) may not be large enough.
    But this is impossible, because the smallest slope (which is negative) that we need to consider is $-2n \cdot 2^{nL} > -2^{3nL}$, so the initial $r = 2^{3nL}$ is in fact large enough.
\end{proof}

\section{Future Research}

Throughout, we have considered a setting where the only decision the agent is able to make is to quit, and the decision to quit is irreversible.  As we argued at the outset, the case where the agent only decides whether to {\em enter} (and this decision is irreversible) leads to the same problem.
However, we could consider richer models where an agent is able to quit, but then has an opportunity to re-enter at certain later times, under certain conditions.

We have also assumed throughout that the agent has no private information.
If the agent has private information, for example about how the agent values different outcomes, we arrive in a dynamic mechanism design context.   As mentioned earlier, in general, in this context we face NP-hardness results~\citep{papadimitriou2016complexity,zhang2021automated}.  Still, we may ask whether the techniques developed in this paper can be generalized to that context, perhaps resulting in polynomial-time algorithms for special cases to which the NP-hardness results do not apply.

One aspect of our approximation of the discounted infinite-horizon case is that it explicitly optimizes only for the first $T$ rounds, and consequently, it might, for example, unsustainably use up all the world's resources by round $T$.  Formally, this is not a problem because, due to the nature of exponential discounting, the remaining rounds are simply not worth much.  Still, one may wonder whether this fails to value long-term sustainability appropriately.  Of course, this issue is not at all unique to our specific setting, but rather a fundamental aspect of exponential discounting.


\bibliographystyle{plainnat}
\bibliography{ref}

\appendix

\end{document}